\documentclass[letterpaper,10 pt, conference]{ieeeconf}
\IEEEoverridecommandlockouts                               
\usepackage{fancyhdr}
\usepackage{hyperref}
\usepackage{xspace}
\usepackage[font={small}]{caption}
\usepackage{color}
\usepackage{setspace}
\usepackage{enumerate}
\usepackage{graphics}
\usepackage{cite}
\usepackage{dsfont}
\usepackage{latexsym}
\usepackage{amsmath}
\usepackage{amssymb}
\usepackage{amsfonts}
\usepackage{amsthm}
\usepackage{times}
\usepackage{sgame}
\usepackage{color}
\usepackage{verbatim}
\usepackage{xcolor}
\usepackage{pgfplots}
\usepackage{epsfig} 
\usepackage{subcaption}

\usepackage{hyperref}
\let\labelindent\relax
\usepackage{enumitem}
\newlist{inparaenum}{enumerate}{2}
\setlist[inparaenum]{nosep}
\setlist[inparaenum,1]{label=\bfseries\arabic*.}
\setlist[inparaenum,2]{label=\arabic{inparaenumi}\emph{\alph*})}

\newcommand{\be}{\begin{equation}}
\newcommand{\ee}{\end{equation}}
\newcommand{\mbb}[1]{\mathbb{#1}}

\newcommand{\mcal}[1]{\mathcal{#1}}

\theoremstyle{plain}
\newtheorem{theorem}{Theorem}[section]

\newtheorem{lemma}{Lemma}[section]

\newtheorem{definition}{Definition}

\def\dSP{\delta_{SP}}
\def\dRT{\delta_{RT}}

\makeatletter
\newcommand*{\rom}[1]{\expandafter\@slowromancap\romannumeral #1@}
\makeatother

\normalsize\title{\LARGE \bf
	An Analysis of Logit Learning with the r-Lambert Function
	\thanks{This work was supported in part by 
Colorado State Bill 18-086, NSF grant  \#ECCS-2346791, and the Netherlands Organization for Scientific Research (NWO-Vici-19902). R. Gavin and M. Cao are with the Faculty of Science and Engineering at the University of Groningen. K. Paarporn is with the Department of Computer Science at University of Colorado, Colorado Springs. Contact: \texttt{\{r.c.gavin,m.cao\}@rug.nl}, \texttt{kpaarpor@uccs.edu}. }  }

\author{
	Rory Gavin, Ming Cao, Keith Paarporn
}

\begin{document}

\maketitle

\begin{abstract}
    
    The well-known replicator equation in evolutionary game theory describes how population-level behaviors change over time when individuals make decisions using simple imitation learning rules. In this paper, we study evolutionary dynamics based on a fundamentally different class of learning rules known as logit learning. Numerous previous studies on logit dynamics provide numerical evidence of bifurcations of multiple fixed points for several types of games. Our results here provide a more explicit analysis of the logit fixed points and their stability properties for the entire class of two-strategy population games -- by way of the $r$-Lambert function. 
    We find that for Prisoner's Dilemma and anti-coordination games, there is only a single fixed point for all rationality levels. However, coordination games exhibit a pitchfork bifurcation: there is a single fixed point in a low-rationality regime, and three fixed points in a high-rationality regime. We provide an implicit characterization for the level of rationality where this bifurcation occurs. In all cases, the set of logit fixed points converges to the full set of Nash equilibria in the high rationality limit.
\end{abstract}

\section{Introduction}

Evolutionary game theory provides a wide range of tools to analyze how large populations of agents behave over time. Originating in the fields evolutionary biology and population dynamics \cite{taylor1978evolutionary,smith1982evolution}, these mathematical tools are now widely applicable to societal and engineered systems like traffic networks, distributed control systems, and social behaviors in epidemics \cite{smith1993new,quijano2017role,weitz2016oscillating,frieswijk2022modeling,gong2022limit,leonardos2023catastrophe,satapathi2023coupled}.  Central to an evolutionary game is the \emph{learning rule}, also called a \emph{revision protocol} \cite{sandholm2010population}, whereby individual agents choose their strategies during the game. When all agents adopt a particular learning rule, one can derive the associated mean dynamical equations that describe how the proportion of agents using a particular strategy changes over time.

Considering that evolutionary game theoretic tools are increasingly used to model human population dynamics and large-scale socio-technical systems, it becomes increasingly important to study the impact of different learning rules on the population-level behavior. Perhaps the most well-studied and utilized protocol is the imitative revision protocol used in the replicator equation. Originally formulated to model biological replication and fitness \cite{taylor1978evolutionary}, the replicator equation is the mean dynamic of agents utilizing imitative learning rules, i.e., modeling strategy selection as the imitation of those who are most successful.

Imitators, by definition, do not assess all their options before making a decision -- they blindly mimic those whom they perceive as more successful.  However, in today's world, advanced technology and an almost inexhaustible amount of information are available to most decision-making agents.  The availability of these resources gives the agents the ability to make informed decisions.  As comparing strategies based on cold hard facts is fundamentally different to parroting those who are more successful, modeling modern decision-making necessitates a revision protocol that includes the dynamics of informed decision-making.

A learning rule that is more aligned with these aspects of decision-making is \emph{logit learning}.  Agents using this protocol weigh the payoffs of choosing one strategy against the payoffs of all other strategies. In other words, each agent decides what to do next using information on the advantages and disadvantages of any given course of action.  The logit learning rule specifies the level at which the agents rationally make this choice using a rationality parameter $\beta \geq 0$, where low values represent choosing at random and high values represent choosing payoff-maximizing actions with high probability.

This paper focuses on the analysis of logit learning dynamics in two-strategy population games. Our analysis centers on completely characterizing the locations and stability of all fixed points for all possible two-strategy games, which includes dominant-strategy, anti-coordination, and coordination games.  While these results are well-known for replicator dynamics, they are not systematically characterized for logit dynamics as its fixed points are described by a transcendental equation.  Consequently, many previous studies of logit dynamics primarily highlight their many interesting properties (e.g. bifurcations and limit cycles) via numerical simulations for specific games \cite{tuyls2003selection,hommes2012multiple,bloembergen2015evolutionary}. One exception is a recent work that analyzes general stability properties of the logit dynamics in population games by using contraction and Lyapunov-based tools \cite{cianfanelli2023logit}. Another recent work studied the impact of logit dynamics in feedback-evolving games, where an environment state co-evolves with agent payoffs \cite{paarporn2024madness}.

Our main contributions in this paper establish an explicit connection between the fixed points of the logit dynamics in two-strategy population games and the $r$-Lambert function \cite{mezHo2017generalization}. By leveraging the properties of this function, we are able to more precisely characterize the behavior of logit fixed points for all $\beta\geq 0$ and for all two-strategy normal-form population games. Of note, we find that only the class of coordination games exhibits a bifurcation: from a single logit fixed point for low rationality, to a set of three fixed points for sufficiently high rationality.

To this end, preliminary background on evolutionary games is provided in Section \ref{sec:prelim}.  Then in Section \ref{sec:ProbDesc}, we give the logit revision protocol and the resultant mean logit dynamics and define the problem.  Sections \ref{sec:analysis} and \ref{sec:Simulations} cover the analysis of the fixed points and present simulations confirming these results, and finally, Section \ref{sec:Conclusion} summarizes our results and future research directions.

\section{Preliminaries}\label{sec:prelim}

We consider a two-strategy normal-form population game with unit mass and label the strategies $\mcal{S} = \{1,2\}$. Let $x_1,x_2 \in [0,1]$ denote the fraction of the population that are using strategies 1 and 2, respectively. As $x_1 + x_2 = 1$, we let $x = x_1$, since it immediately follows that $x_2 = 1-x$. We term $x \in [0,1]$ the \emph{population state}. 

The payoff matrix is given by
\begin{equation}
    A = 
    \begin{bmatrix}
        R & S \\ T & P
    \end{bmatrix}
\end{equation}
with the four entries $R,S,T,P \in \mbb{R}$. An agent using strategy 1 experiences a payoff of $R$ when it encounters another agent using strategy 1, and a payoff of $S$ when it encounters an agent using strategy 2. Likewise, an agent using strategy 2 experiences a payoff of $T$ and $P$ when encountering an agent using strategy 1 and 2, respectively.

We concisely represent the payoffs to each strategy as
\begin{equation}
    \begin{aligned}
	   \pi_1(x) &= [A [x,1-x]^\top ]_1 = Rx + S(1-x), \\
        \pi_2(x) &= [A [x,1-x]^\top ]_2 = Tx + P(1-x),
    \end{aligned}
\end{equation}
where $[M]_i$ denotes the $i$th row of matrix $M$.

\begin{definition}
    A population state $x$ is a \emph{Nash equilibrium} if and only if it satisfies one of the following:

    \begin{enumerate}
        \item If $x = 0$, then $\pi_2(0) > \pi_1(0)$;
        \item If $x \in (0,1)$, then $\pi_1(x) = \pi_2(x)$; or
        \item If $x=1$, then $\pi_1(0) > \pi_2(0)$.
    \end{enumerate}
\end{definition}

The collection of all 2-strategy normal form games can be classified into four distinct types. Defining the parameters $\dSP := S - P$ and $\dRT := R - T$, the four types correspond to the four quadrants in the parameter space $(\delta_{SP},\delta_{RT})\in\mathbb{R}^2$.  Throughout this paper, we shall denote these quadrants as
\begin{align*}
\mathcal{Q}_{\text{\rom{1}}} &:= \left\{ (\delta_{SP}, \delta_{RT}) \; \middle| \; (\delta_{SP}, \delta_{RT}) \in \mathbb{R}_0^+ \times \mathbb{R}_0^+\right\} \setminus \left\{ (0,0) \right\}\\
\mathcal{Q}_{\text{\rom{2}}} &:= \left\{ (\delta_{SP}, \delta_{RT}) \; \middle| \; (\delta_{SP}, \delta_{RT}) \in \mathbb{R}^- \times \mathbb{R}^+\right\} \\
\mathcal{Q}_{\text{\rom{3}}} &:= \left\{ (\delta_{SP}, \delta_{RT}) \; \middle| \; (\delta_{SP}, \delta_{RT}) \in \mathbb{R}_0^- \times \mathbb{R}_0^-\right\} \setminus \left\{ (0,0) \right\} \\
\mathcal{Q}_{\text{\rom{4}}} &:= \left\{ (\delta_{SP}, \delta_{RT}) \; \middle| \; (\delta_{SP}, \delta_{RT}) \in \mathbb{R}^+ \times \mathbb{R}^-\right\},
\end{align*}

where $\mathbb{R}^+$ is the set of all positive real numbers, $\mathbb{R}^-$ is the set of all negative real numbers, and the subscript $0$ in $\mathbb{R}_0^+$ and $\mathbb{R}_0^-$ denote $\mathbb{R}^+ \cup \{0\}$ and $\mathbb{R}^- \cup \{0\}$, respectively.

Quadrant \rom{1} corresponds to games where strategy 1 is dominant whose unique Nash equilibrium is $x=1$. Quadrant \rom{2} describes coordination games which have three Nash equilibria at $x = 0$, $x = \frac{\dSP}{\dSP - \dRT}$, and $x = 1$. Quadrant \rom{3} describes the Prisoner's Dilemma which has a unique Nash equilibrium $x=0$. Quadrant \rom{4} describes anti-coordination games whose unique Nash equilibrium is $x = \frac{\dSP}{\dSP - \dRT}$.

\subsection{The replicator equation}

Agents in the population must dynamically revise their strategy choices over time. They do so using a \emph{revision protocol}, which models how they decide whether or not to switch from their current strategy.  A revision protocol is given by the collection of functions
\begin{equation}
    \rho_{ij}(x), \ \forall i,j \in \mcal{S}
\end{equation}
where $\rho_{ij}(x) \geq 0$ quantifies the rate at which $i$-strategists switch to strategy $j$. Specification of the revision protocol allows one to express the \emph{mean dynamics},
\begin{equation}\label{eq:mean_dyn}
    \dot{x} = (1-x) \rho_{21}(x) - x \rho_{12}(x)
\end{equation}
which describes the rate at which the fraction of 1-strategists changes in the population. In many evolutionary game analyses, the agents follow an imitative revision protocol,
\begin{equation}
    \rho_{ij}(x) = x_j[\pi_j(x) - \pi_i(x)]_+
\end{equation}
where $[a]_+ = \max\{a,0\}$. Here, an $i$-strategist encounters a $j$-strategist in the population with probability $x_j$, and will switch to $j$ at a rate proportional to the magnitude of the difference in their payoffs. The imitative protocol induces the mean dynamics \eqref{eq:mean_dyn}
\begin{equation}\label{eq:replicator}
    \dot{x} = x(1-x)g(x)
\end{equation}
where
\begin{equation}\label{eq:g}
    g(x) := \pi_1(x) - \pi_2(x) = \dRT x + \dSP (1-x)
\end{equation}
is the payoff difference between strategy 1 and 2. The equation \eqref{eq:replicator} is called the \emph{replicator equation}.

\section{Problem Description}\label{sec:ProbDesc}
In this paper, however, we focus on a different revision protocol based on a boundedly-rational learning rule known as the logit protocol:
\begin{equation}\label{eq:logit_protocol}
    \rho_{ij}(x) = \rho_{j}(x) = \frac{\exp(\beta \pi_j(x))}{\sum_{k\in\mcal{S}} \exp(\beta \pi_k(x))}.
\end{equation}
where $\beta \geq 0$ is the \emph{rationality level}. 

The logit protocol is fundamentally different from the imitation protocol. In \eqref{eq:logit_protocol}, the agent weighs all possible strategies with more weight over ones that give higher payoffs -- consistent with informed or technologically-assisted decision-making. The degree to which higher-payoff strategies are favored depends on the rationality level. For $\beta = 0$, agents just choose any strategy uniformly at random, and for large $\beta$, the logit protocol reflects a best-response decision.

The expression for the mean dynamics \eqref{eq:mean_dyn} induced by the logit protocol is given by
\begin{equation}\label{eq:logit_dyn}
    \dot{x} = f(x)  = \frac{e^{\beta g(x)}}{1+e^{\beta g(x)}} - x
\end{equation}
and is referred to as the \emph{logit dynamics}.

In the next section, we detail the connection between the $r$-Lambert function and the fixed points of \eqref{eq:logit_dyn}. 

\section{Analysis}\label{sec:analysis}

One can easily locate the fixed points for each game for large $\beta$ by evaluating the limit of $f(x^*)=0$ as $\beta$ tends to infinity.  Hence, one arrives at the Theorem \ref{thm:AsympOfFPs}.

\begin{theorem}[Fixed Point Values as $\beta \to \infty$]\label{thm:AsympOfFPs}
The fixed points of \eqref{eq:logit_dyn} converge to the Nash equilibria for all $(\delta_{SP},\delta_{RT})\in\mathbb{R}^2\setminus \left\{(\delta_{SP},\delta_{RT}) \; \middle| \; \delta_{RT}=\delta_{SP} \right\}$ as $\beta$ tends to infinity.
\end{theorem}
\begin{proof}
See Appendix \ref{ap:asympFP}
\end{proof}

Theorem \ref{thm:AsympOfFPs} immediately reveals an interesting property without having to solve the transcental equation $f(x^*)=0$.  For two-strategy games using logit learning, the set of fixed points, $\mathcal{S}^*$, defined \[
\mathcal{S}^* := \left\{ x\in[0,1] \; \middle| \; f(x) = 0 \right\},
\] for any particular game converges to the full set of Nash equilibria.  In the context of two-strategy population games, this provides a strengthened version of the results from \cite{cianfanelli2023logit} which state that the set of fixed points for an arbitrary $n$-strategy game must be a subset of the Nash equilibria.

This method, however, is limited to infinite values of $\beta$.  For an analysis of finite $\beta \geq 0$, one must solve the transcendental equation $f(x^*)=0$.  To do so, however, one requires the $r$-Lambert Function \cite{mezHo2017generalization}.

\subsection{Fixed Points Location in Terms of $r$-Lambert Function }\label{sec:r-LamForm}
\par The $r$-Lambert function is the solution $y=W_r(z)$ to equations of the form $ye^y + ry = z$ for some $r,z\in\mathbb{R}$ \cite{mezHo2017generalization}.  Define $k:= -\beta m$, $m := \delta_{RT} - \delta_{SP}$, and $r:=e^{\beta \delta_{SP}}$.  At a fixed point $x^*$ of \eqref{eq:logit_dyn}, $f(x^*)=0$ is given in terms of $k$, $m$, and $r$ by
\begin{equation*}
   kr = kx^*e^{kx^*} + rkx^*.
\end{equation*}

This expression satisfies the form necessary for a $r$-Lambert solution. Therefore, \eqref{eq:logit_dyn_sol} is the explicit expression for the fixed points of \eqref{eq:logit_dyn} when $k\neq 0$. 
\begin{equation}\label{eq:logit_dyn_sol}
x^* = \frac{1}{k}W_r(kr)
\end{equation}

 If $k=0$, i.e., if $\beta=0$ or $\delta_{RT} = \delta_{SP}$, an algebraic expression exists for sole the fixed point, given by
\begin{equation}\label{eq:logit_dyn_sol_k_eq_0}
    x^* = \frac{r}{1+r}.
\end{equation}

As such, one acquires an explicit expression for the fixed points $x^*$ for any level of rationality $\beta$ and any game $(\delta_{SP},\delta_{RT})$.

\subsection{Quantity of Fixed Points}\label{sec:noOfFPs}
Though $x^*$ has a single algebraic solution whenever $k=0$, $x^*$ has no straightforward algebraic solution whenever $k \neq 0$.  Despite this, one can deduce the quantity of solutions and their stability using the properties of the $r$-Lambert function.

For $y=W_r(z)$, the number of solutions to the $r$-Lambert function depends upon the values of $r$ and $z$ \cite[Theorem 4]{mezHo2017generalization}.  Therefore, the number of solutions to \eqref{eq:logit_dyn_sol} depends on $r$ and $kr$, which in turn depend on the agents' level of rationality and the type of game.  Examining $r$ and $kr$ over all $\beta$ and $(\delta_{SP},\delta_{RT})$, one finds the number of fixed points per parameter regime, summarized in Theorem \ref{thm:NoOfFPs}. 
\begin{theorem}[The Number of Fixed Points]\label{thm:NoOfFPs}
\par Let $|\mathcal{X}|$ denote the cardinality of set $\mathcal{X}$.  The number of fixed points depends on the quadrant in which the game $(\delta_{SP},\delta_{RT})$ resides.
\begin{itemize}
\item If $(\delta_{SP},\delta_{RT}) \in \mathcal{Q}_{\text{\rom{1}}} \cup \mathcal{Q}_{\text{\rom{3}}} \cup \mathcal{Q}_{\text{\rom{4}}}$, then $\lvert \mathcal{S}^* \rvert = 1$ for all $\beta \geq 0 $;
\item If $(\delta_{SP},\delta_{RT}) \in \mathcal{Q}_{\text{\rom{2}}}$ and $\delta_{RT} \neq -\delta_{SP}$, then for all $\beta \geq 0$,
\begin{itemize}
\item $\lvert \mathcal{S}^* \rvert = 1$ for all $\beta < \beta_r$;
\item $\lvert \mathcal{S}^* \rvert = 2$ for all $\beta = \beta_r$; and
\item $\lvert \mathcal{S}^* \rvert = 3$ for all $\beta > \beta_r$;
\end{itemize}
where $\beta_r$ is defined implicitly for  $\delta_{RT} > -\delta_{SP}$ as \[
W_{0}(-re) = 1 - \frac{\beta_{r} m}{2} + \frac{\sqrt{\beta_{r} m\left(\beta_{r} m - 4\right)}}{2},
\] or for $\delta_{RT} < -\delta_{SP}$ as \[
W_{-1}(-re) = 1 - \frac{\beta_{r} m}{2} - \frac{\sqrt{\beta_{r} m\left(\beta_{r} m - 4\right)}}{2};
\] and $W_0(z)$ and $W_{-1}(z)$ are the 0th and --1st branches of the Lambert $W$ function.
\item If $(\delta_{SP},\delta_{RT}) \in \mathcal{Q}_{\text{\rom{2}}}$ and $\delta_{RT} = -\delta_{SP}$, then for all $\beta \geq 0$,
\begin{itemize}
\item $\lvert \mathcal{S}^* \rvert = 1$ for $\beta \leq -\frac{2}{\delta_{SP}}$; and
\item $\lvert \mathcal{S}^* \rvert = 3$ for $\beta > -\frac{2}{\delta_{SP}}$.
\end{itemize}
\end{itemize}
\end{theorem}

\begin{proof}
See Appendix \ref{ap:NoOfFPs}
\end{proof}

Theorem \ref{thm:NoOfFPs} concretely establishes the number of fixed points $x^*$ for any $\beta \geq 0$.  Furthermore, given a finite $\beta_{r}$ for $(\delta_{SP},\delta_{RT}) \in \mathcal{Q}_{\text{\rom{2}}}$, Theorem \ref{thm:NoOfFPs} proves that the number of fixed points strictly increases for increasing $\beta$.  More generally, however, as $\beta \to \infty$, this theorem neatly coincides with Theorem \ref{thm:AsympOfFPs} for all types of games.  As such, Theorems \ref{thm:AsympOfFPs} and \ref{thm:NoOfFPs} combined characterize the quantity of fixed points $x^*$ for any level of rationality $\beta \geq 0$ and all games $(\delta_{SP},\delta_{RT}) \in \mathbb{R}^2$.



\subsection{Stability of Fixed Points}\label{sec:FPsStability}
With the location of any fixed point $x^*$ analytically shown using the $r$-Lambert function in Section \ref{sec:r-LamForm} and the quantity thereof proven in Section \ref{sec:noOfFPs}, each fixed point's stability characteristics are all that's left outstanding.  Fortunately, each fixed point's stability as $\beta$ tends to infinity is easily found by taking the limit of $\frac{df}{dx}$ as $\beta$ tends to infinity.

\begin{theorem}[Fixed Point Stability as $\beta \to \infty$]\label{thm:AsympFPStab}
As $\beta$ tends to infinity and for any $(\delta_{SP},\delta_{RT})\in\mathbb{R}^2$, all fixed points are stable except the fixed point $x^*=\frac{\delta_{SP}}{\delta_{SP}-\delta_{RT}}$ in the parameter regime $(\delta_{SP},\delta_{RT})\in\mathcal{Q}_{\text{\rom{2}}}$.
\end{theorem}
\begin{proof}
See Appendix \ref{ap:AsympFPStab}
\end{proof}

Therefore, on top of the fixed points $x^*$ tending to the full set of Nash equilibria by Theorem \ref{thm:AsympOfFPs}, the fixed points of the logit dynamics have the same stability characteristics for large $\beta$ as the fixed points of replicator equation.  As logit learning with high $\beta$ represents a best-response revision protocol, this result demonstrates that the best-response coincides with imitating others despite the imitative and logit learning rules working in fundamentally different ways.

Furthermore, using the $r$-Lambert function, one acquires the general condition for any fixed point's stability expressed in Theorem \ref{thm:GenFPStab}
\begin{theorem}[Fixed Point Stability for Finite $\beta$]\label{thm:GenFPStab}
The fixed point $x^*$ is
\begin{enumerate}
    \item stable for all $\beta > 0$ if $\delta_{SP} = \delta_{RT}$;
    \item stable for all $\delta_{SP} \neq \delta_{RT}$ if $\beta = 0$; and
    \item stable for some $\beta\neq0$ and $\delta_{SP}\neq\delta_{RT}$ if 
    \begin{equation}\label{eq:genFPstab}
        \frac{1}{k}W_r^2(kr) - W_r(kr) - 1 \leq 0
    \end{equation}
\end{enumerate}
\end{theorem}
\begin{proof}
See Appendix \ref{ap:FPStabCond}
\end{proof}
This theorem highlights again that the $r$-Lambert function relates yet another fundamental trait of the logit dynamics: its stability.  The following section checks this and the previous results numerically for each type of game and different levels of rationality.

\section{Simulations}\label{sec:Simulations}
Simulations confirm the theorems and lemmas presented herein and in the appendices.  Figure \ref{fig:FPsByParamReg} illustrates representative plots of $x^*$ values in all four quadrants of the parameter space for increasing $\beta$ given constant parameters $(\delta_{SP},\delta_{RT})$.  All simulations confirm Theorems \ref{thm:AsympOfFPs}, \ref{thm:NoOfFPs}, and \ref{thm:AsympFPStab}, as illustrated in Figure \ref{fig:FPsByParamReg}.

Furthermore, simulations suggest that there exist further properties of this system.  As illustrated in Figure \ref{fig:FPsByParamReg}, all simulations suggest that the fixed points tend \emph{monotonically} to the Nash equilibria for each game type. On top of this, numerically computing \eqref{eq:genFPstab} suggests that all fixed points are stable except the one that approaches the mixed Nash equilibrium for coordination games.  In other words, increasingly better responses tend to the replicator equation's response in terms of fixed point location and stability.  As such, analytical characterizations of monotonicity and stability will be the subjects of future investigation.

\begin{figure}[h]
    \centering
    \includegraphics[width=1\linewidth]{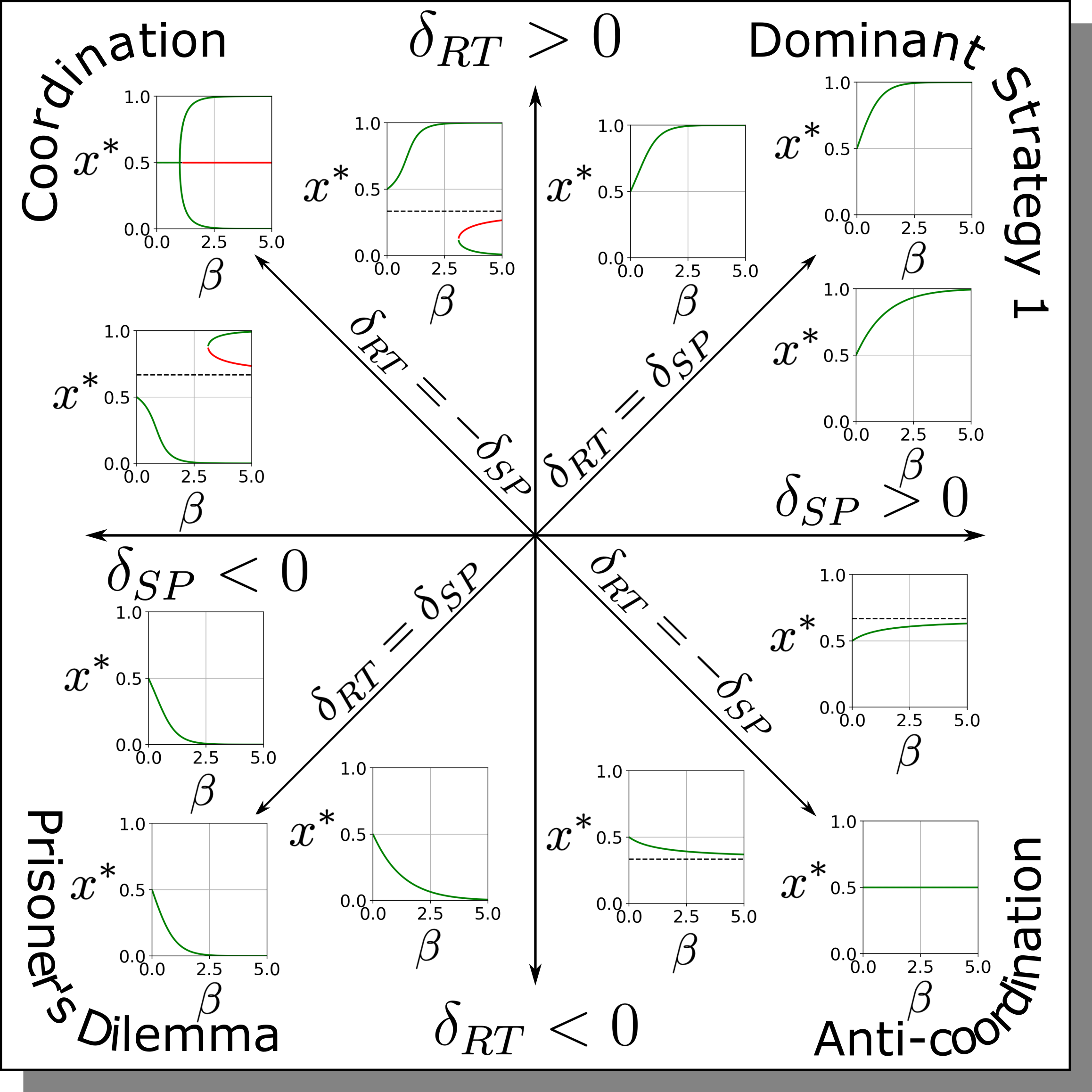}
    \caption{Plots of the fixed points $x^*$ of \eqref{eq:logit_dyn} as a function of $\beta$ for points in the parameter space $(\delta_{SP}, \delta_{RT}) = (\pm 1, \pm 2)$, $(\pm 2, \pm 1)$, and $(\pm 2, \pm 2)$.  Dashed lines denote mixed Nash equilibria.  This diagram plots stable values of $x^*$ in green and unstable ones in red, as calculated numerically using \eqref{eq:genFPstab}.  Gaps in the plots indicate the limits of the numerical methods used to compute the solutions to \eqref{eq:logit_dyn_sol}.}\label{fig:FPsByParamReg}
    \label{fig:enter-label}
\end{figure}

\section{Conclusions}\label{sec:Conclusion}
This research analytically characterized evolutionary game theoretic models using the logit revision protocol with the $r$-Lambert function.  We showed that at low rationality levels only coordination games exhibit bifurcations.  At very high rationality levels, however, logit learning exhibits the same dynamical outcomes, both in terms of fixed point location and stability as imitation learning despite belonging to a fundamentally different class of revision protocol.

Future research will endeavor to demonstrate analytically that the fixed points $x^*$ tend monotonically towards the Nash equilibria for each game and to analytically characterize the stability for where $k \neq 0$.  This research would directly influence any future research into the control of the logit dynamics using the rationality parameter, which would have potential applications towards industrial robotics with limited resources like time, energy, and materials.

\bibliographystyle{IEEEtran}
\bibliography{library}

\appendix

%

\subsection{Proof of Theorem \ref{thm:AsympOfFPs}}\label{ap:asympFP}
To find the number of solutions for large $\beta$ entails only taking the limit of $f(x^*)=0$ as $\beta$ tends to infinity.

From section \ref{sec:r-LamForm}, one knows that for $k=0$  $x^* = \frac{r}{1+r}$. Therefore, for $(\delta_{SP}, \delta_{RT}) = (0,0)$, $x^* = \frac{1}{2}$.

Next, consider $(\delta_{SP},\delta_{RT})\in\mathcal{Q}_{\text{\rom{1}}}$.  To find the value of $x^*$ for increasing values of $\beta$, take the limit
\begin{equation}\label{eq:betatoinf}
1 = x^* + x^*\lim_{\beta \to \infty}e^{-\beta(mx^*+\delta_{SP})}.
\end{equation}

Define $\mu := mx^*+\delta_{SP}$.  For all $(\delta_{SP},\delta_{RT}) \in \mathcal{Q}_{\text{\rom{1}}}$, $m \geq -\delta_{SP}$.  Therefore, $\mu \geq 0$ for any $x^*\in[0,1]$.  As such, $-\beta \mu \leq 0$  for all $(\delta_{SP},\delta_{RT}) \in \mathcal{Q}_{\text{\rom{1}}}$ and $\beta \geq 0$.

If $m>-\delta_{SP}$, the limit evaluates to \[
1 = x^* + x^*\lim_{\beta \to \infty} e^{-\beta\mu} \implies 1 = x^* + 0^+.
\]

If $m=-\delta_{SP}$, the limit simplifies to \[
1 = x^* + x^*\lim_{\beta \to \infty} e^{-\beta\delta_{SP}(1-x^*)} \implies 1 = x^* + 0^+.
\]

Therefore, $x^* \to 1$ from below as $\beta \to \infty$ for all $(\delta_{SP},\delta_{RT}) \in \mathcal{Q}_{\text{\rom{1}}}$ and all $\beta \geq 0$.

Next, consider $(\delta_{SP},\delta_{RT})\in\mathcal{Q}_{\text{\rom{2}}}$.  Rearranging the limit \eqref{eq:betatoinf}  
one acquires 
\begin{equation}\label{eq:altbetatoinf}
\frac{1-x^*}{x^*} = \lim_{\beta \to +\infty}\left(e^{x^* + \frac{\delta_{SP}}{m}}\right)^{-\beta m }
\end{equation}

The value of $m$ is greater than zero for all $(\delta_{SP},\delta_{RT})\in\mathcal{Q}_{\text{\rom{2}}}$.  Hence, the limit tends to:
\begin{enumerate}
\item $+\infty$ if $x^*<\frac{\delta_{SP}}{\delta_{SP}-\delta_{RT}}$;
\item $0^+$ if $x^*>\frac{\delta_{SP}}{\delta_{SP}-\delta_{RT}}$; or
\item $+1$ if $x^*=\frac{\delta_{SP}}{\delta_{SP}-\delta_{RT}}$.
\end{enumerate}

Checking these three cases, one finds that $x^*$ approaches $1$, $0$, and $\frac{\delta_{SP}}{\delta_{SP}-\delta_{RT}}$ as $\beta$ tends to infinity for parameters in the region $(\delta_{SP},\delta_{RT})\in\mathcal{Q}_{\text{\rom{2}}}$.

Next, consider $(\delta_{SP},\delta_{RT})\in\mathcal{Q}_{\text{\rom{3}}}$.  To find the value of $x^*$ for increasing values of $\beta$, take the limit \eqref{eq:betatoinf} in this parameter regime.

Again, define $\mu := mx^*+\delta_{SP}$.  For all $(\delta_{SP},\delta_{RT}) \in \mathcal{Q}_{\text{\rom{3}}}$, $m \leq -\delta_{SP}$.  Therefore, $\mu \leq 0$ for any $x\in[0,1]$.  As such, $-\beta \mu \geq 0$ for all $(\delta_{SP},\delta_{RT}) \in \mathcal{Q}_{\text{\rom{3}}}$ and $\beta \geq 0$.

Dividing both sides of \eqref{eq:betatoinf} by $e^{-\beta (mx+\delta_{SP0})}$ and taking the limit of the resulting expression, one finds that $x^* \to 0$ from above as $\beta \to \infty$ for all $(\delta_{SP},\delta_{RT}) \in \mathcal{Q}_{\text{\rom{3}}}$ and all $\beta \geq 0$.

Next, consider $(\delta_{SP},\delta_{RT})\in\mathcal{Q}_{\text{\rom{4}}}$.  Take the natural logarithm of \eqref{eq:altbetatoinf}, which yields \[
\ln(1-x^*) - \ln(x^*)= \lim_{\beta \to \infty} \left( -\beta m \left(x^* + \frac{\delta_{SP}}{m}\right) \right).
\]

Dividing both sides by $-\beta m$ yields\[
\lim_{\beta\to\infty} \left(\frac{\ln(1-x^*)}{-\beta m} - \frac{\ln(x^*)}{-\beta m}\right) =\lim_{\beta\to\infty}\left( x^* + \frac{\delta_{SP}}{m}\right).
\]

Finally, by evaluating the above limit, one finds that the fixed point $x^* \to \frac{\delta_{SP}}{\delta_{SP}-\delta_{RT}}$ as $\beta \to \infty$ for all $(\delta_{SP},\delta_{RT}) \in \mathcal{Q}_{\text{\rom{4}}}$ and all $\beta \geq 0$,\footnote{Furthermore, one can also show this using \cite[Theorem 8]{mezHo2017generalization}, i.e., \[
\lim_{\beta\to\infty} x^* = \lim_{\beta\to\infty} \frac{1}{k} \left(\ln(kr) + \ln\left( \frac{1}{\ln(kr)} - \frac{r}{kr}\right) \right) = -\frac{\delta_{SP}}{\delta_{RT}-\delta_{SP}}
\]} thereby completing the proof of Theorem \ref{thm:AsympOfFPs}.

\subsection{Proof of Theorem \ref{thm:NoOfFPs}}\label{ap:NoOfFPs}
Using the results found by Mez\H o and Baricz \cite[Theorem 4]{mezHo2017generalization}, the number of solutions to the $r$-Lambert function depends on the values of $r$ and $kr$.  Define $\alpha_{-1} := W_{-1}(-re) - 1$, $\alpha_{0} := W_{0}(-re) - 1$, and $f_r(z) = ze^z + rz$ as in \cite[Theorem 4]{mezHo2017generalization}.

\begin{lemma}\label{lma:NoOfSolns}
    The conditions for certain quantities of fixed points where $k \neq 0$ are as follows.
\begin{itemize}
\item If $r \geq e^{-2}$, then $\lvert \mathcal{S}^* \rvert = 1$;
\item If $0 < r < e^{-2}$, then
    \begin{itemize}
        \item $\lvert \mathcal{S}^*\rvert = 1$ if $kr > f_r(\alpha_{-1})$ or $kr < f_r(\alpha_{0})$;
        \item $\lvert \mathcal{S}^*\rvert = 2$ if $kr = f_r(\alpha_{-1})$ or $kr = f_r(\alpha_{0})$; and
        \item $\lvert \mathcal{S}^*\rvert = 3$ if $f_r(\alpha_{0}) < kr < f_r(\alpha_{-1})$: and
    \end{itemize}
\item If $r < 0$, then
    \begin{itemize}
        \item $\lvert \mathcal{S}^*\rvert = 1$ if $kr = f_r(\alpha_{0})$; and
        \item $\lvert \mathcal{S}^*\rvert = 2$ if $kr > f_r(\alpha_{0})$.
    \end{itemize}
\end{itemize}
\end{lemma}

As $r = e^{\beta \delta_{SP}}$, $r$ is never less than or equal to zero for all $\beta \geq 0$ and any $\delta_{SP}\in\mathbb{R}$.  Solving for when $r \geq e^{-2}$ yields
\begin{equation*}
r \geq e^{-2} \implies \begin{cases} \beta \geq -\frac{2}{\delta_{SP}} & \delta_{SP}>0 \\
\beta \leq -\frac{2}{\delta_{SP}} & \delta_{SP}<0
\end{cases}
\end{equation*}

If $\delta_{SP}>0$ then $-\frac{2}{\delta_{SP}}<0$.  As $\beta$ must be greater than zero by definition, for any model with $\delta_{SP}>0$, there will be only one solution.  Hence, $\lvert \mathcal{S}^*\rvert = 1$ for all $(\delta_{SP},\delta_{RT})\in\mathcal{Q}_{\text{\rom{1}}}\cup\mathcal{Q}_{\text{\rom{4}}}$.

On the other hand, for models with $(\delta_{SP},\delta_{RT})\in\mathcal{Q}_{\text{\rom{2}}}\cup\mathcal{Q}_{\text{\rom{3}}}$, $\delta_{SP}<0$ and $-\frac{2}{\delta_{SP}}>0$.  By Lemma \ref{lma:NoOfSolns}, there is an interval $[0,-\frac{2}{\delta_{SP}}]$ wherein there is guaranteed to be only one solution, but for all $\beta > -\frac{2}{\delta_{SP}}$, there are one, two, or three solutions.

To find when there are one, two, or three fixed points for $(\delta_{SP},\delta_{RT})\in\mathcal{Q}_{\text{\rom{2}}}\cup\mathcal{Q}_{\text{\rom{3}}}$, one can rearrange $kr = f_r(\alpha_{i})$ to acquire
\begin{equation*}
W_{i}^2(-re) + (\beta m - 2)W_{i}(-re) + 1 = 0
\end{equation*}
for $i=0,-1$, where $i$ represents a particular branch of the Lambert $W$ function.

Define
\begin{equation*}
h_i(\beta):=W_{i}^2(-re) + (\beta m - 2)W_{i}(-re) + 1
\end{equation*}
for $i=0,-1$ whose domain is the interval $(-\frac{2}{\delta_{SP}},\infty)$.  With this definition, one can express the requirements of Lemma \ref{lma:NoOfSolns} as those specific to the logit dynamics, summarized in Lemma \ref{lma:NoOfSolnsUpdated}.

\begin{lemma}[Quantity of Logit Dynamics Fixed Points]\label{lma:NoOfSolnsUpdated} 
For all $(\delta_{SP},\delta_{RT}) \in \mathcal{Q}_{\text{\rom{1}}}\cup\mathcal{Q}_{\text{\rom{4}}}$, $\lvert \mathcal{S}^*\rvert = 1$. For all $(\delta_{SP},\delta_{RT}) \in \mathcal{Q}_{\text{\rom{2}}}\cup\mathcal{Q}_{\text{\rom{3}}}$,
\begin{itemize}
\item If $\beta \leq -\frac{2}{\delta_{SP}}$, then $\lvert \mathcal{S}^* \rvert = 1$;
\item If $\beta > -\frac{2}{\delta_{SP}}$, then
    \begin{itemize}
        \item $\lvert \mathcal{S}^*\rvert = 1$ if $h_{-1}(\beta) > 0$ or $h_{0}(\beta) < 0$;
        \item $\lvert \mathcal{S}^*\rvert = 2$ if $h_{-1}(\beta) = 0$ or $h_{0}(\beta) = 0$; and
        \item $\lvert \mathcal{S}^*\rvert = 3$ if $h_{-1}(\beta) < 0$ and $h_{0}(\beta) > 0$.
    \end{itemize}
\end{itemize} 
\end{lemma}

Without the ability to directly ascertain the stationary points of the function $h_i(\beta)$ which is quadratic in $W_i(-re)$ another way to demonstrate which intervals of $\beta$ satisfy the conditions in Lemma \ref{lma:NoOfSolnsUpdated} is to check $h_i(\beta)$'s stationary points and roots.  These analyses, given in Appendices \ref{ap:h0h1sps} and \ref{ap:h0h1rts}, are summarized by Lemmas \ref{lma:h0h1sps} and \ref{lma:h0h1rts}.

\begin{table}[ht]
\begin{center}
\begin{tabular}{c c c c}
Parameter Regime & Limit & $h_0(\beta)$ & $h_{-1}(\beta)$ \\
\hline
$\delta_{RT} > -\delta_{SP}$ & $\beta \to -\frac{2}{\delta_{SP}}$ & neg. & neg. \\
$\delta_{RT} > -\delta_{SP}$ & $\beta \to \infty$ & pos. & neg. \\
$\delta_{RT} = -\delta_{SP}$ & $\beta \to -\frac{2}{\delta_{SP}}$ & 0 & neg. \\
$\delta_{RT} = -\delta_{SP}$ & $\beta \to \infty$ & 0 & neg.\\
$0 < \delta_{RT} < -\delta_{SP}$ & $\beta \to -\frac{2}{\delta_{SP}}$& pos. & pos.\\
$0 < \delta_{RT} < -\delta_{SP}$ & $\beta \to \infty$& pos. & neg. \\
$\delta_{RT} \leq 0$ & $\beta \to -\frac{2}{\delta_{SP}}$& pos.& pos. \\
$\delta_{RT} \leq 0$ & $\beta \to \infty$& pos. & pos.\\
\end{tabular}
\caption{Table listing the sign of the values of $h_0(\beta)$ and $h_{-1}(\beta)$ at the limits of their domain $(-\frac{2}{\delta_{SP}},\infty)$ for parameters $(\delta_{SP},\delta_{RT}) \in\mathcal{Q}_{\text{\rom{2}}}\cup \mathcal{Q}_{\text{\rom{3}}}$.}\label{tab:h0h1sign}
\end{center}
\end{table}

Combining the knowledge of Lemmas \ref{lma:h0h1sps} and \ref{lma:h0h1rts} along with $h_0(\beta)$ and $h_{-1}(\beta)$'s signs at the extremes of their domains listed in Table \ref{tab:h0h1sign}, one can deduce the results of Theorem \ref{thm:NoOfFPs}.

For $(\delta_{SP},\delta_{RT}) \in \mathcal{Q}_{\text{\rom{2}}}\cap\left\{(\delta_{SP},\delta_{RT}) \; \middle| \; \delta_{RT} > -\delta_{SP} \right\}$, $h_0(\beta)$ and $h_{-1}(\beta)$ have no stationary points by Lemma \ref{lma:h0h1sps} and only $h_0(\beta)$ has one root by Lemma \ref{lma:h0h1rts}.  As $h_{-1}(\beta)$ has no roots and is negative at both extremes of its domain, $h_{-1}(\beta)$ is less than zero over its entire domain. The function $h_{0}(\beta)$ has one root, has no stationary points, starts negative-valued at the beginning of its domain, and tends to be positive valued towards the end of its domain.  As such, one knows that:
\begin{itemize}
    \item $h_{0}(\beta)<0$ for $\beta$ such that $-\frac{2}{\delta_{SP}} < \beta < \beta_{r,0}$;
    \item  $h_0(\beta)=0$ when $\beta = \beta_{r,0}$; and
    \item $h_0(\beta)>0$ for all $\beta > \beta_{r,0}$,
\end{itemize}
for $\beta_{r,0}$ defined as in Lemma \ref{lma:h0h1rts}.
Therefore, by Lemma \ref{lma:NoOfSolnsUpdated}, one finds that 
\begin{itemize}
    \item $\lvert \mathcal{S}^* \rvert = 1$ for $\beta < \beta_{r,0}$ as $h_0(\beta)<0$ and $h_{-1}(\beta)<0$;
    \item $\lvert \mathcal{S}^* \rvert = 2$ for $\beta = \beta_{r,0}$ as $h_0(\beta)=0$ and $h_{-1}(\beta)<0$; and 
    \item $\lvert \mathcal{S}^* \rvert = 3$ for $\beta > \beta_{r,0}$ as $h_0(\beta)>0$ and $h_{-1}(\beta)<0$.
\end{itemize}

Using a similar deductive method for the rest of the parameter space $\mathcal{Q}_{\text{\rom{2}}}\cup \mathcal{Q}_{\text{\rom{3}}}$, one confirms the results of Theorem \ref{thm:NoOfFPs}, thereby proving Theorem \ref{thm:NoOfFPs}.

\subsection{Proof of Theorem \ref{thm:AsympFPStab}}\label{ap:AsympFPStab}
Checking the stability of \eqref{eq:logit_dyn} entails checking $\left. \frac{df}{dt}\right\rvert_{x^*} \leq 0$.  By Lemma \ref{lma:FPStabCond}, as $\beta\to\infty$, this amounts to showing the following holds true.
\begin{equation*}
\lim_{\beta\to\infty} \frac{1}{k}W_r^2(kr) -W_r(kr) - 1 \leq 0
\end{equation*}

Factoring this expression yields
\begin{equation*}
\lim_{\beta\to\infty} W_r(kr)\left(\frac{1}{k}W_r(kr) -1\right) - 1 \leq 0.
\end{equation*}

As $x^* = \frac{1}{k}W_r(kr)$, one acquires
\begin{equation*}
\lim_{\beta\to\infty} kx^*\left(x^* -1\right) - 1 \leq 0.
\end{equation*}

As the values of $x^*$ are known as $\beta$ tends to infinity, one needs only to substitute each $x^*$ along with its corresponding the parameters $(\delta_{SP},\delta_{RT})$ into this expression and evaluate the limit.  By evaluating this limit for each fixed point $x^*$ in each quadrant of the parameter space, one finds that only the limit for the fixed point $x^* = \frac{\delta_{SP}}{\delta_{SP}-\delta_{RT}}$ for parameters $(\delta_{SP},\delta_{RT})\in\mathcal{Q}_{\text{\rom{2}}}$ evaluates to a value greater than zero.

\subsection{The Stationary Points of $h_0(\beta)$ and $h_{-1}(\beta)$}\label{ap:h0h1sps}
The impossibility of algebraically finding $\beta$  such that $h_i(\beta)=0$ necessitates an indirect analysis of the function $h_i(\beta)$.  This appendix details the analysis of $h_i(\beta)$'s stationary points for $i=0,-1$.  The results are summarized in Lemma \ref{lma:h0h1sps}.
\begin{lemma}[Stationary Points of $h_i(\beta)$]\label{lma:h0h1sps}
For games with parameters $(\delta_{SP},\delta_{RT})\in \mathcal{Q}_{\text{\rom{2}}} \cup \mathcal{Q}_{\text{\rom{3}}}$, the function
\begin{itemize}
    \item $h_0(\beta)$ only has stationary points in the parameter regime $\delta_{SP} < \delta_{RT} < -\delta_{SP}$.  In this parameter regime, $h_0(\beta)$ has only one stationary point.
    \item $h_{-1}(\beta)$ only has stationary points in the parameter regime $0 < \delta_{RT} < -\delta_{SP}$. In this parameter regime, $h_{-1}(\beta)$ has only one stationary point.
\end{itemize}
\end{lemma}

\begin{proof}
The derivative of $h_i(\beta)$ is
\begin{align*}
\frac{dh_i}{d\beta} &= \frac{W_{i}^2\left(-re\right)(2\delta_{SP}+m)}{1+W_{i}\left(-re\right)}\\ 
				&+ \frac{W_{i}\left(-re\right)(\delta_{SP}(\beta m-2)+m)}{1+W_{i}\left(-re\right)}.
\end{align*}

Hence, the function $h_i(\beta)$ has a stationary point at $\beta_{s}$, implicitly defined by 
\begin{align*}
W_{i}(-re) =& 0, \\
W_{i}(-re) =& \frac{- \beta_s \delta_{SP} m + 2 \delta_{SP} - m}{2 \delta_{SP} + m}.
\end{align*}

By Lemma \ref{lma:MonoAndAsympW0W1}, neither $W_0(-re)$ nor $W_{-1}(-re)$ will ever equal zero.  Hence, there can only be one stationary point, if any.

To analyze this stationary point, define
\begin{equation*}
\gamma(\beta) := \frac{- \beta \delta_{SP} m + 2 \delta_{SP} - m}{2 \delta_{SP} + m}.
\end{equation*}

For there to be a stationary point $\beta_s$, the functions $W_i(-re)$ and $\gamma(\beta)$ must intersect at $\beta = \beta_s$.  By inspection, $\gamma(\beta)$ is linear with a slope of $-\frac{\delta_{SP} m}{2 \delta_{SP} + m}$.  Furthermore, $\gamma(-\frac{2}{\delta_{SP}}) = 1$.  Therefore, $\gamma(\beta)$ is greater than $W_i(-e^{\beta \delta_{SP} +1})$ for $i=0,-1$ near $-\frac{2}{\delta_{SP}}$.

By Lemma \ref{lma:MonoAndAsympW0W1},  $W_0(\beta)$ and $W_{-1}(\beta)$ are monotonic.  The former's slope tends from $+\infty$ to $0$, and the latter's slope tends from $-\infty$ to $\delta_{SP}$. As $\gamma(\beta)$ is greater than $W_i(-re)$ near $-\frac{2}{\delta_{SP}}$ for $i=0,-1$, $\gamma(\beta)$'s slope must therefore be negative for it to intersect $W_0(-re)$.  Furthermore, $\gamma(\beta)$'s slope must less than $\delta_{SP}$ for it to intersect $W_{-1}(-re)$.

Considering only $(\delta_{SP},\delta_{RT}) \in \mathcal{Q}_{\text{\rom{2}}} \cup \mathcal{Q}_{\text{\rom{3}}}$, $\gamma(\beta)$'s slope is negative for all $\delta_{RT}$ such that $\delta_{SP} < \delta_{RT} < -\delta_{SP}$.  Furthermore, $\gamma(\beta)$'s slope is less than $\delta_{SP}$ for $\delta_{RT} > 0$.  Therefore, 
\begin{itemize}
    \item $h_0(\beta)$ has one stationary point for all $(\delta_{SP},\delta_{RT}) \in \mathcal{Q}_{\text{\rom{2}}} \cup \mathcal{Q}_{\text{\rom{3}}}$ for $\delta_{RT}$ such that $\delta_{SP} < \delta_{RT} < -\delta_{SP}$;
    \item $h_{-1}(\beta)$ has one stationary point for all $(\delta_{SP},\delta_{RT}) \in \mathcal{Q}_{\text{\rom{2}}} \cup \mathcal{Q}_{\text{\rom{3}}}$ for $\delta_{RT}$ such that $0 < \delta_{RT} < -\delta_{SP}$; and
    \item $h_0(\beta)$ and $h_{-1}(\beta)$ have no stationary points for all $(\delta_{SP},\delta_{RT}) \in \mathcal{Q}_{\text{\rom{2}}} \cup \mathcal{Q}_{\text{\rom{3}}}$ for $\delta_{RT}$ such that $\delta_{RT} \leq \delta_{SP}$ and $\delta_{RT} \geq -\delta_{SP}$.
\end{itemize}

To show that the stationary point of $h_0(\beta)$ is always positive, note that $W_0(-re)\in (-1,0)$.  As the existence of stationary points is guaranteed for $\delta_{SP} < \delta_{RT} < -\delta_{SP}$, $\gamma(\beta)$ must too be in the interval $(-1,0)$ for some $\beta > -\frac{2}{\delta_{SP}}$.  Solving the inequality $-1 < \gamma(\beta) < 0$ demonstrates that there is a stationary point $\beta_s$ such that
\begin{equation*}
\frac{2\delta_{SP} - m}{\delta_{SP}m} < \beta_s < \frac{4}{m}.
\end{equation*}

Note that $-\frac{2}{\delta_{SP}}<\frac{2\delta_{SP} - m}{\delta_{SP}m}$ for all $\delta_{SP} < \delta_{RT} < -\delta_{SP}$.  The value of $h_i(\beta)$ evaluated at $\beta_s$ is
\begin{equation*}
h_i(\beta) = - \frac{m \left(\beta_s m - 4\right) \left(\beta_s \delta_{SP0} \left(\delta_{SP0} + m\right) + m\right)}{\left(2 \delta_{SP0} + m\right)^{2}}
\end{equation*}

The denominator of this expression is always positive.  Furthermore, as $\beta_s < \frac{4}{m}$ and $m > 0$ for all $\delta_{SP} < \delta_{RT} < -\delta_{SP}$, the term $-m(\beta_s m - 4)$ is alway positive too.  Therefore, the sign of $h_i(\beta)$ depends on the sign of $\beta_s \delta_{SP0} \left(\delta_{SP0} + m\right) + m$.

To prove that $h_0(\beta_s)>0$ by contradiction, assume that $\beta_s \delta_{SP0} \left(\delta_{SP0} + m\right) + m<0$.  If this is true, then so must be
\begin{equation*}
\beta_s > -\frac{m}{\delta_{SP0} \left(\delta_{SP0} + m\right)} \ \text{if $\delta_{RT}>0$}, \text{ and}
\end{equation*}
\begin{equation*}
\beta_s < -\frac{m}{\delta_{SP0} \left(\delta_{SP0} + m\right)} \ \text{if $\delta_{RT}<0$}.
\end{equation*}

As there is as stationary point $\beta_s \in (\frac{2\delta_{SP} - m}{\delta_{SP}m},\frac{4}{m})$, for these inequalities to be true, so must be
\begin{equation*}
-\frac{m}{\delta_{SP0} \left(\delta_{SP0} + m\right)} < \frac{4}{m} \ \text{if $\delta_{RT}>0$}, \text{ and}
\end{equation*}
\begin{equation*}
-\frac{m}{\delta_{SP0} \left(\delta_{SP0} + m\right)} > \frac{2\delta_{SP} - m}{\delta_{SP}m} \ \text{if $\delta_{RT}<0$}.
\end{equation*}

However, these inequalities imply that
\begin{equation*}
0 > (m+2\delta_{SP})^2 \ \text{if $\delta_{RT}>0$}, \text{ and}
\end{equation*}
\begin{equation*}
\delta_{SP}(m+2\delta_{SP}) < 0 \ \text{if $\delta_{RT}<0$},
\end{equation*}
which are impossible for $\delta_{SP} < \delta_{RT} < -\delta_{SP}$ and $\delta_{SP}<0$.  Therefore, $\beta_s \delta_{SP0} \left(\delta_{SP0} + m\right) + m>0$, thus proving that $h_0(\beta_s)>0$ for all $(\delta_{SP},\delta_{RT}) \in \mathcal{Q}_{\text{\rom{2}}} \cup \mathcal{Q}_{\text{\rom{3}}}$ for $\delta_{RT}$ such that $\delta_{SP} < \delta_{RT} < -\delta_{SP}$.

To show that the stationary point of $h_{-1}(\beta)$ is always greater than zero, note that the limit of $h_{-1}(\beta)$'s slope as $\beta$ tends to $-\frac{2}{\delta_{SP}}$ is positive.  As by Table \ref{tab:h0h1sign} $h_{-1}(\beta)$ is positive at the beginning and negative at the end of its domain, if its slope and value are both positive at the start of its domain, then the stationary point must be a positive global maximum of $h_{-1}(\beta)$.  Thus, one proves Lemma \ref{lma:h0h1sps} in full.
\end{proof}

\subsection{The Roots of $h_0(\beta)$ and $h_{-1}(\beta)$}\label{ap:h0h1rts}
Having already examined $h_i(\beta)$'s stationary points in Appendix \ref{ap:h0h1sps}, we now turn to indirectly characterizing this function's roots.  The results of this characterization are summarized in Lemma \ref{lma:h0h1rts}.
\begin{lemma}[Roots of $h_i(\beta)$]\label{lma:h0h1rts}
The functions $h_0(\beta)$ and $h_{-1}(\beta)$ only have roots for $(\delta_{SP},\delta_{RT}) \in \mathcal{Q}_{\text{\rom{2}}}$.  For any $(\delta_{SP},\delta_{RT}) \in \mathcal{Q}_{\text{\rom{2}}}$,
\begin{itemize}
    \item If $\delta_{RT} > -\delta_{SP}$, then
    \begin{itemize}
        \item $h_{-1}(\beta)$ has no roots; and
        \item $h_0(\beta)$ has one root $\beta_{r,0} > -\frac{2}{\delta_{SP}}$ given implicitly by \[
W_0(-re) = 1 - \frac{\beta_{r,0} m}{2} + \frac{\sqrt{\beta_{r,0} m\left(\beta_{r,0} m - 4\right)}}{2};
\]
    \end{itemize}
    \item If $0 < \delta_{RT} < -\delta_{SP}$, then
    \begin{itemize}
        \item $h_{0}(\beta)$ has no roots; and
        \item $h_{0}(\beta)$ has one root given implicitly by \[
W_{-1}(-re) = 1 - \frac{\beta_{r,-1} m}{2} - \frac{\sqrt{\beta_{r,-1} m\left(\beta_{r,-1} m - 4\right)}}{2};
\] and
    \end{itemize}
    \item If $\delta_{RT} = -\delta_{SP}$ or $\delta_{RT} \leq 0$, then $h_{0}(\beta)$ and $h_{-1}(\beta)$ have no roots.
\end{itemize}
\end{lemma}

\begin{proof}
As the function $h_i(\beta)$ is quadratic in $W_i(-re)$, $\beta_r$ is given implicitly by
\begin{equation*}
W_i(-re) = 1 - \frac{\beta_r m}{2} \pm \frac{\sqrt{\beta_r m\left(\beta_r m - 4\right)}}{2}.
\end{equation*}

Note that for all $\beta_r > \frac{4}{m}$,
\begin{align*}
1 - \frac{\beta_r m}{2} + \frac{\sqrt{\beta_r m\left(\beta_r m - 4\right)}}{2} &\in (-1,0), \\
1 - \frac{\beta_r m}{2} - \frac{\sqrt{\beta_r m\left(\beta_r m - 4\right)}}{2} &\in (-\infty,-1).
\end{align*}

Due to these expressions' ranges and as $W_0(-re)\in(-1,0)$ and $W_{-1}(-re)\in(-\infty,-1)$, $h_0(\beta)$'s only possible solution is
\begin{equation*}
W_0(-re) = 1 - \frac{\beta_{r,0} m}{2} + \frac{\sqrt{\beta_{r,0} m\left(\beta_{r,0} m - 4\right)}}{2},
\end{equation*}
and $h_{-1}(\beta)$'s only possible solution is
\begin{equation*}
W_{-1}(-re) = 1 - \frac{\beta_{r,-1} m}{2} - \frac{\sqrt{\beta_{r,-1} m\left(\beta_{r,-1} m - 4\right)}}{2},
\end{equation*}
where $\beta_{r,0},\beta_{r,-1} > \frac{4}{m}$.

Note that the existence of these implicit expressions for the roots of $h_0(\beta)$ and $h_{-1}(\beta)$ imply that these functions can have at most one root.  Knowing this result, the number of stationary points from Lemma \ref{lma:h0h1sps}, and the signs of $h_0(\beta)$ and $h_{-1}(\beta)$ at the extremes of their domains from Table \ref{tab:h0h1sign}, one can prove when $h_0(\beta)$ and $h_{-1}(\beta)$ have a root.

For $\delta_{RT}> -\delta_{SP}$, neither $h_0(\beta)$ nor $h_{-1}(\beta)$ have any stationary points by Lemma \ref{lma:h0h1sps}. As by Table \ref{tab:h0h1sign} $h_{-1}(\beta)$ is negative at both ends of its domain, it is therefore negative over its entire domain. As $h_0(\beta)$ transitions from negative values at the beginning of its domain to positive values the end of its domain, $h_0(\beta)$ must cross zero exactly once at $\beta_{r,0}$.  Therefore, $h_{0}(\beta)\leq0$ for $\beta \in (-\frac{2}{\delta_{SP}}, \beta_{r,0}]$ and $h_{0}(\beta)>0$ for $\beta \in (\beta_{r,0},\infty)$.

At $\delta_{RT} = -\delta_{SP}$, neither $h_0(\beta)$ nor $h_{-1}(\beta)$ have stationary points by Lemma \ref{lma:h0h1sps}.  Furthermore, these functions equal zero as beta tends to $ -\frac{2}{\delta_{SP}}$.  As $h_0(\beta)>0$ and $h_{-1}(\beta)<0$ for all $\beta > -\frac{2}{\delta_{SP}}$ by Table \ref{tab:h0h1sign}, these functions have no roots when $\delta_{RT} = -\delta_{SP}$.

For $0 < \delta_{RT} < -\delta_{SP}$, only $h_{-1}(\beta)$ has a stationary point by Lemma \ref{lma:h0h1sps}.  As $h_0(\beta)$ is greater than zero at both ends of its domain by Table \ref{tab:h0h1sign}, $h_0(\beta) > 0$ over its entire domain.  As $h_{-1}(\beta)$ can have at most one root and as $h_{-1}(\beta)$ transitions from positive values at the beginning of its domain to negative values the end of its domain, $h_{-1}(\beta)$ must cross zero exactly once at $\beta_{r,-1}$.  Therefore, $h_{-1}(\beta)\geq0$ for $\beta \in (-\frac{2}{\delta_{SP}}, \beta_{r,-1}]$ and $h_{-1}(\beta)<0$ for $\beta \in (\beta_{r,-1},\infty)$.

For $\delta_{RT} \leq 0$, as $h_{-1}(\beta)$ has no stationary points by Lemma \ref{lma:h0h1sps} and as $h_{-1}(\beta)>0$ at both ends of its domain, $h_{-1}(\beta)$ has no roots for $\delta_{RT} \leq 0$.  Although $h_0(\beta)$ has a stationary point, its stationary point is positive by Lemma \ref{lma:h0h1sps}.  As $h_0(\beta) >0$ at the extremes of its domain, $h_0(\beta)$ has no roots for any $\delta_{RT} \leq 0$.

Therefore, proves Lemma \ref{lma:h0h1rts}, where
\begin{itemize}
\item For $\delta_{RT}> -\delta_{SP}$, $h_0(\beta)$ has one root and $h_{-1}(\beta)$ has none;
\item For $\delta_{RT} = -\delta_{SP}$, neither $h_0(\beta)$ nor $h_{-1}(\beta)$ have any roots;
\item For $0<\delta_{RT} < -\delta_{SP}$, $h_{-1}(\beta)$ has one root and $h_{0}(\beta)$ has none; and
\item For $\delta_{RT} \leq 0$, neither $h_0(\beta)$ nor $h_{-1}(\beta)$ have any roots.
\end{itemize}
\end{proof}

\subsection{Monotonicity of $W_i(-e^{\beta \delta_{SP} + 1})$}\label{ap:MonoAndAsympW0W1}
The following Lemma is necessary to prove the existence and number of stationary points of $h_i(\beta)$. 
\begin{lemma}[Monotonicity of $W_i(-e^{\beta \delta_{SP} + 1})$]\label{lma:MonoAndAsympW0W1}
For $(\delta_{SP},\delta_{RT})\in\mathcal{Q}_{\text{\rom{2}}}\cup \mathcal{Q}_{\text{\rom{3}}}$ and for all $\beta > -\frac{2}{\delta_{SP}}$,
\begin{itemize}
    \item $W_0(-e^{\beta \delta_{SP} +1})$ is strictly increasing; and
    \item $W_{-1}(-e^{\beta \delta_{SP} +1})$ is strictly decreasing.
\end{itemize}
\end{lemma}

\begin{proof}
The ranges of $W_0(z)\in(-1,0)$ and $W_{-1}(z)\in(-\infty,-1)$ for some $z\in\mathbb{R}$ follow from \cite{corless1996lambert}.  Taking the derivative of $W_i(-e^{\beta \delta_{SP} +1})$,
\begin{equation}\label{eq:dWi_dbeta}
W'_i(-re) = \frac{-\delta_{SP}reW_i(-re)}{-re(1+W_i(-re))} = \delta_{SP} \frac{1}{1 + \frac{1}{W_i(-re)}}
\end{equation}

Given the ranges $W_0(z)\in(-1,0)$ and $W_{-1}(z)\in(-\infty,-1)$, one finds that $W_0(-re)$ strictly increases and $W_{-1}(-re)$ strictly decreases over the interval $\beta \in (-\frac{2}{\delta_{SP}},\infty)$.
\end{proof}

\subsection{Condition for the Fixed Point Stability}\label{ap:FPStabCond}
The following Lemma demonstrates how the stability of logit dynamics fixed points are intrinsically linked to the $r$-Lambert function.
\begin{lemma}[Condition for Fixed Point Stability]\label{lma:FPStabCond}
A fixed point $x^*=\frac{1}{k}W_r(kr)$ is stable for $k \neq 0$ if
\begin{equation*}
\frac{1}{k}W_r^2(kr) -W_r(kr) - 1 \leq 0.
\end{equation*}
\end{lemma}

\begin{proof}
The derivative $\left. \frac{df}{dx} \right|_{x^*}$ in terms of $k=-\beta m$, $m=\delta_{RT} - \delta_{SP}$, and $r = e^{\beta \delta_{SP}}$ is

\begin{equation*}
    \left. \frac{df}{dx} \right|_{x^*} = \frac{-k e^{W_r(kr)}r^{-1}}{\left(1 + e^{W_r(kr)}r^{-1}\right)^{2}} - 1.
\end{equation*}

By the definition of the $r$-Lambert function, one knows that $W_r(kr)e^{W_r(kr)} + rW_r(kr)=kr$.  Therefore,
\begin{equation*}
    e^{W_r(kr)}r^{-1} = \frac{k}{W_r(kr)} - 1.
\end{equation*}
Substituting this expression into that of $\left. \frac{df}{dx} \right|_{x^*}$ and simplifying yields
\begin{equation*}
\left. \frac{df}{dx} \right|_{x^*} = \frac{1}{k}W_r^2(kr) -W_r(kr) - 1
\end{equation*}

Therefore, all fixed points $x^*=\frac{1}{k}W_r(kr)$ where $k\neq0$ are stable if
\begin{equation*}
    \frac{1}{k}W_r^2(kr) -W_r(kr) - 1 \leq 0.
\end{equation*}
\end{proof}


\end{document}